\documentclass[journal]{IEEEtran}

\usepackage[latin1]{inputenc}
\usepackage{graphicx}
\usepackage{amssymb,amsmath,amsfonts,amsthm}
\usepackage{algorithm,algorithmic}
\usepackage{cite}
\usepackage{float}
\usepackage{graphics}
\usepackage{subfigure}
\usepackage{xspace}
\usepackage[usenames,dvipsnames]{color}
\usepackage{epsfig,}
\usepackage{dsfont}

\newtheorem{theorem}{Theorem}

\newcommand{\mfrac}[2]{\genfrac{}{}{0pt}{}{#1}{#2}}

\begin{document}

\title{Estimation-Energy Tradeoff for Scalar Gauss-Markov Signals with Kalman Filtering}

\author{Ioannis Krikidis, \IEEEmembership{Fellow, IEEE} and Constantinos Psomas, \IEEEmembership{Senior Member, IEEE}\vspace*{-5mm}
\thanks{I. Krikidis and C. Psomas are with the Department of Electrical and Computer Engineering, University of Cyprus, Cyprus (email: \{krikidis, psomas\}@ucy.ac.cy). This work has received funding from the European Research Council (ERC) under the European Union's Horizon 2020 research and innovation programme (Grant agreement No. 819819).}}

\maketitle

\begin{abstract}
In this letter, we investigate a receiver architecture, which uses the received signal in order to simultaneously harvest energy and estimate a Gauss-Markov linear process. We study three communication scenarios: i) static channel, ii) Rayleigh block-fading channel, and iii) high power amplifier (HPA) nonlinearities at the transmitter side. Theoretical results for the minimum mean square error as well as the average harvested energy are given for all cases and the fundamental tradeoff between estimation quality and harvested energy is characterized. We show that channel fading improves the estimation performance while HPA requires an extended Kalman filter at the receiver and significantly affects both the estimation and the harvesting efficiency.
\end{abstract}
\vspace{-0.2cm}

\vspace{-0.6cm}
\section{Introduction}

Wireless power transfer (WPT) is an enabling technology for future cyber-physical systems, where a massive number of devices (e.g., sensors, RFID tags, actuators, etc) with low-rate and low-power requirements are interconnected in order to monitor and control critical infrastructures. It can provide energy sustainability, battery-less implementation and therefore it seems as an essential technology for the next generation of cyber-physical systems that require continuous estimation/monitoring of a physical process. The integration of WPT in wireless communication systems has been studied extensively over the last years and the fundamental tradeoff between information and energy transfer has been studied from information theoretic \cite{VAR,KRI2}, signal processing and networking perspectives \cite{CLE}. However, the interplay of WPT with conventional parameter/signal estimation, which is the key task in the above systems, is still an unexplored area.

Recent works study the integration of WPT in signal estimation problems. The authors in \cite{HON} investigate a distributed estimation problem, where sensors harvest energy from radio-frequency signals and use it, in order to transmit their observations to a fusion center; the parameters of WPT are designed to minimize the mean-square error (MSE) of the final estimate. This work has been extended in \cite{VEN} for scenarios with multiple-antenna sensor nodes, while the work in \cite{KRI3} studies the problem from a macroscopic point-of-view by taking into account the spatial randomness of the sensor network. All the aforementioned works decouple the harvesting operation from the signal estimation process and energy harvesting is used in order to allow sensors to communicate their observations. Another topology with practical interest is the simultaneous parameter estimation and energy harvesting from the same signal; in accordance with the simultaneous information and power transfer concept, various practical architectures such as the power-splitting (PS) technique can be used for this purpose \cite{CLE}.  To the best of the authors' knowledge, this dual use of the received signal (signal/state estimation, energy harvesting) is still an open problem in the literature.

In this paper, we investigate a PS-based architecture that uses the received signal in order to simultaneously estimate a scalar Gauss-Markov signal (i.e.,  sequence of states evolving as a dynamic linear model) \cite{PAR1,PAR2} and harvest energy. Specifically, we study the use of first-order Gauss-Markov signals for WPT. We consider three scenarios: i) transmission over a static wireless channel, ii) transmission over a Rayleigh block-fading channel, and iii) the case where the transmitter is affected by high power amplifier (HPA) nonlinearities \cite{KRI2}. Closed-form expressions for the average harvested energy (linear model) are derived for both finite and asymptotic regimes.  A receiver architecture is also introduced, which enables a simultaneous estimation of the transmitted signal and energy harvesting, by employing a PS technique. For the estimation process, the receiver applies a linear Kalman filter on the corresponding fraction of the signal by taking into account the memory of the scalar Gauss-Markov signal; for the scenario with HPA, we employ an extended Kalman filter. Theoretical results are given for the minimum MSE (MMSE) for all the considered scenarios. The considered architecture is characterized by a fundamental tradeoff between quality of the estimation (MMSE) and average harvested energy.
\noindent {\it Notation:}  $\mathbb{E}\{\cdot\}$ denotes statistical expectation, $\exp(\lambda)$ denotes the exponential distribution with rate parameter $\lambda$, $\mathcal{CN}(\mu,\sigma^2)$ represents the complex Gaussian distribution with mean $\mu$ and variance $\sigma^2$,  $\Re\{\cdot\}$ and $\Im\{\cdot\}$ indicate real and imaginary operators, respectively, and $j=\sqrt{-1}$.

\vspace{-0.2cm}
\section{Estimation-energy tradeoff}

In this section, we introduce the considered receiver architecture and study the fundamental estimation-energy tradeoff for three scenarios: i) static channel, ii) Rayleigh block-fading, and iii) transmission with HPA non-linearities. 

\vspace{-0.3cm}
\subsection{Gauss-Markov over a static channel}

We assume a scalar complex Gauss-Markov signal model\cite{KAY}, given by\footnote{The coefficient $a$ depends on the specific application (e.g., navigation, robotics, etc.) and determines the correlation between consecutive measurements; in practical dynamic systems where small fluctuations occur, we assume $|a|<1$, which results in stable systems. The case $a > 1$ corresponds to unstable open-loop processes where, in most cases, they are associated with a control system to ensure stability.}
\begin{align}
\text{State:}\;\; &x(n)=a x(n-1)+u(n), \;\;n\geq 0, \\
\text{Observations:}\;\; &y(n)=hx(n)+v(n),
\end{align}
where $n$ is the time index $x(-1)\sim \mathcal{CN}(\mu_0,\sigma_0^2)$,  $u(n)\sim \mathcal{CN}(0,\sigma_u^2)$, $v(n)\sim \mathcal{CN}(0,\sigma_v^2)$, $|a|<1$,  and $h \in \mathbb{C}$ denotes the complex channel coefficient which is assumed constant and known at the receiver.
According to the principles of PS \cite{CLE}, the receiver splits the received signal $y(n)$ into two components by using a splitting factor $\rho$ i.e., $\sqrt{\rho}y(n)$ is converted to the baseband and is used to estimate $x(n)$ and $\sqrt{1-\rho} y(n)$ is driven towards the energy harvesting branch. During the baseband conversion, an additional circuit noise, $q(n)$, is present which is modeled as additive white Gaussian noise (AWGN) with zero mean and variance $\sigma_q^2$. The system model considered is depicted in Fig. \ref{fig1} and could refer to wireless sensor networks that monitor the state of a dynamic physical phenomenon (e.g., air quality monitoring, object tracking, etc.). The obtained measurements are sent over a fading channel by using analog forwarding to avoid processing delays \cite{PAR2}. This setup seems interesting for future cyberphysical systems, where a massive number of sensors/devices exchange measurements and are characterized by strict computation/delay/energy constraints.

For data estimation, the equivalent (baseband) observation can be written as 
\begin{align}
&y'(n)=\sqrt{\rho}h x(n)+e(n),
\end{align}
where $e(n)=\sqrt{\rho}v(n)+q(n)$ is circularly symmetric complex Gaussian with variance $\sigma_e^2=\rho \sigma_v^2+\sigma_q^2$. Due to the data correlation, the estimation of $x(n)$ takes into account the whole history of observations and is based on the Kalman filter equations\footnote{ The Kalman filter is a time-varying linear filter and achieves linear MMSE in a sequential way by introducing one new sample per time and by applying trivial (low-complexity) algebraic computations \cite{KAY}.} \cite{KAY, PAR1} i.e.,
\begin{align}
&\hat{x}(n|n-1)=a \hat{x}(n-1|n-1), \label{e1}\\
&M(n|n-1)=a^2 M(n-1|n-1)+\sigma_u^2, \label{e2} \\
&K(n)=\frac{M(n|n-1)\sqrt{\rho}h^*}{\sigma_e^2+\rho |h|^2 M(n|n-1)}, \label{e22} \\
&\hat{x}(n|n)=\hat{x}(n|n-1)+K(n)[y'(n)-\sqrt{\rho}h \hat{x}(n|n-1)], \label{e3}\\
&M(n|n)=(1-K(n)\sqrt{\rho}h)M(n|n-1) \nonumber \\
&\;\;\;\;\;\;\;\;\;\;\;\;=\frac{a^2 \sigma_e^2 M(n-1|n-1)+\sigma_e^2 \sigma_u^2}{\sigma_e^2+\rho |h|^2(a^2 M(n-1|n-1)+\sigma_u^2)}, \label{e4}
\end{align}
where $\hat{x}(n|n-1)=\mathbb{E}\{x(n)|y(0), y(1),\ldots, y(n-1)\}$ is the MMSE prediction based on the previous estimated state, \eqref{e2} represents the minimum prediction MSE, \eqref{e22} gives the Kalman gain, \eqref{e3} denotes the correction equations and \eqref{e4} is the updated MMSE. The initial conditions are given by $\hat{x}(-1|-1)=\mu_0$ and $M(-1|-1)=\sigma_0^2$.

For the asymptotic case ($n\rightarrow \infty$), the MMSE converges to a deterministic value which can be calculated by solving the following quadrature equation \cite{KAY}
\begin{align}
M(\infty)=\frac{\sigma_e^2(a^2 M(\infty)+\sigma_u^2)}{\sigma_e^2+\rho|h|^2 (a^2 M(\infty)+\sigma_u^2)},
\end{align}
with a solution equal to
\begin{align}
M(\infty)=\frac{-Q_2+\sqrt{Q_2^2-4Q_1Q_3}}{2Q_3},\label{asymp1}
\end{align}
where $Q_1=\rho|h|^2a^2$,  $Q_2=\rho|h|^2\sigma_u^2+\sigma_e^2(1-a^2)$,  and $Q_3=-\sigma_u^2\sigma_e^2$.
For the special case with $\rho=0$ i.e., the received signal is used only for energy harvesting, the MMSE becomes equal to
\begin{align}
M(\infty)=\sigma_u^2/(1-a^2)=\; \sigma^2.
\end{align}
Then, by using a linear WPT model\footnote{A linear WPT model refers to the linear operation regime of the rectenna and serves as a useful guideline for practical implementations; more sophisticated non-linear models can be also used \cite{CLE, KRI2}.} and by neglecting energy harvesting from the AWGN component, the average harvested energy for the time instant $n$ becomes equal to
\begin{align}
&\mathcal{E}(n)=\zeta \mathbb{E}\{ |\sqrt{1-\rho} h x(n)|^2\}=\zeta (1-\rho)|h|^2\mathbb{E}\{|x(n)|^2\},\label{harvesting}
\end{align}
where $\zeta$ is the conversion efficiency of the energy harvester and $|x(n)|$ is a Rician distributed random variable with parameters $a^{n+1}\mu_0$ and $\sqrt{(a^{2n+2}\sigma_0^2+\sigma^2(1-a^{2n+2}))/2}$, since $x(n) \sim \mathcal{CN}(a^{n+1}\mu_0, a^{2n+2}\sigma_0^2 + \sigma^2(1-a^{2n+2}))$. Therefore, we have
\begin{align}
\mathbb{E}\{|x(n)|^2\} = a^{2n+2}(\mu_0^2 + \sigma_0^2)+\sigma^2(1-a^{2n+2}).\label{var2}
\end{align}
Then, for the asymptotic case with $n\rightarrow \infty$, we have $x(n)\sim \mathcal{CN}(0,\sigma^2)$ and thus $|x(n)|$ follows a Rayleigh distribution with 
\begin{align}
&\mathbb{E}\{|x(n)|^2\}\rightarrow \sigma^2,\label{var3}
\end{align}
and so
\begin{align}\label{no_hpa}
&\mathcal{E}(\infty)=\lim_{n\rightarrow \infty}\mathcal{E}(n)=\zeta \frac{(1-\rho)|h|^2 \sigma_u^2}{1-a^2}.
\end{align}
It is worth noting that $\mathcal{E}(\infty)=\zeta |h|^2\mathcal{M}(\infty)$ for $\rho=0$.

\begin{figure}
\includegraphics[width=\linewidth]{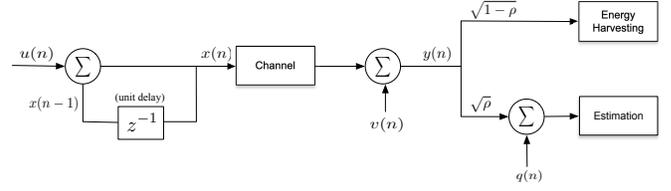}
\vspace{-0.6cm}
\caption{System model; A receiver architecture that simultaneously estimates a Gauss-Markov signal and harvests energy by using PS.} \label{fig1}
\end{figure}

\vspace{-0.2cm}
\subsection{Gauss-Markov over a Rayleigh fading channel}

In this case, the state $x(n)$ is transmitted through a Rayleigh fading channel with $|h(n)|^2 \sim \exp(-\lambda)$ \cite{PAR1,PAR2}. The Kalman filter equations are similar to \eqref{e1}-\eqref{e4} by replacing the constant $|h|^2$ with the random variable $|h(n)|^2$. The MMSE becomes a nonstationary random variable with a cumulative density function (CDF) upper/lower bounded (for any $n$) according to Theorem \ref{th1} given below. 
\begin{theorem}\label{th1}
For a Rayleigh fading channel, the MMSE is a nonstationary random variable with a CDF bounded as follows
\begin{align}
F \left(\sigma_u^2(1+a^2),x\right) \leq F_{M_{n}}(x)\leq F \left(\frac{\sigma_u^2}{1-a^2},x\right),
\end{align}
where
\begin{align}
F(c,x)=\left\{ \begin{array}{ll} \exp \left(-\frac{\lambda \sigma_e^2}{\rho}\left(\frac{1}{x}-\frac{1}{c} \right)\right),\;\;\textrm{if}\;x<c, \\
1,\;\;\textrm{elsewhere}.\end{array}\right.\label{cdf1}
\end{align}
\end{theorem}
\begin{proof}
The proof is given in Appendix A.
\end{proof}
It is obvious that for small values of $a$, the two bounds become tight and both approximate efficiently the actual CDF i.e., $F_{M_{n}}(x)\approx F(\sigma_u^2, x)$.

Regarding the bounding of the mean value of the MMSE, we derive the expected values associated with the two considered bounds. Specifically, the expected value corresponding to the CDF in \eqref{cdf1} is given by
\begin{align}
\Theta(c)&=\int_{0}^{c}(1- F(c,x))dx \nonumber \\
&=\frac{\lambda \sigma_e^2}{\rho}\exp\left(\frac{\lambda \sigma_e^2}{\rho c} \right) \Gamma\left(0,\frac{\lambda \sigma_e^2}{\rho c} \right),
\end{align} 
which follows through integration by parts and from the definition of the upper incomplete gamma function $\Gamma(s,x)=\int_x^{\infty}t^{s-1}\exp(-t)dt$. Therefore, the mean MMSE is bounded as follows i.e., 
\begin{align}
\Theta(\sigma_u^2(1+a^2))\leq \mathbb{E}\{M_{n}\}\leq \Theta\left(\frac{\sigma_u^2}{1-a^2} \right),
\end{align}
and thus for low values of $a$, we have $\mathbb{E}\{M_{n}\} \approx \Theta(\sigma_u^2)$. By using the expression in \eqref{harvesting}, the mean harvested energy becomes
\begin{align}
&\mathcal{E}(n)=\zeta \mathbb{E}\{|\sqrt{1-\rho}h(n)x(n)|^2 \} \nonumber \\
&=\zeta (1-\rho)\mathbb{E}\{|h(n)|^2\}\mathbb{E}\{|x(n)|^2 \}=\frac{\zeta(1-\rho)}{\lambda}\mathbb{E}\{|x(n)|^2\},
\end{align}
where $\mathbb{E}\{|x(n)|^2\}$ is given by \eqref{var2} or \eqref{var3}.  It is worth noting that the above mathematical framework is general and can be applied to other channel fading models e.g., Rice/Nakagami.

\subsection{High power amplifier nonlinearities}

We study the impact of high power amplifier nonlinearities on the estimation-energy tradeoff. 
Specifically, we assume that the transmitter suffers from HPA nonlinearities \cite{KRI2}; in this case the transformation of the states to observations becomes nonlinear. The output of the power amplifier for a complex input $z$ is given by 
\begin{align}
F_{\text{HPA}}(z)&=f_A(|z|)\exp(j(f_P(|z|)+\arg{z})),
\end{align}
where $f_A(|z|)$ and $f_P(|z|)$ denote the AM/AM and AM/PM distortion, respectively. For the solid state power amplifier (SSPA) model \cite{KRI2}, the AM/AM and AM/PM characteristic functions are given by
\begin{align}\label{char_fun}
f_A(|z|)=\frac{|z|}{\left[ \left( \frac{|z|}{A_{\text{sat}}} \right)^{2\beta}+1 \right]^{\frac{1}{2\beta}}},\;\;\;f_P(|z|)=0,
\end{align} 
where $A_{\text{sat}}$ is the output saturation voltage and $\beta$ represents the smoothness of the transition from the linear regime to the saturation. In this case, the Gauss-Markov nonlinear model is
\begin{align}
\text{State:}\;\;&x(n)=a x(n-1))+u(n), \;\;n\geq 0, \\
\text{Observations:}\;\;&y(n)=hF_{\text{HPA}}(x(n))+v(n).
\end{align}
By linearizing the function $F_\text{HPA}(x(n))$ through the first-order Taylor expansion i.e., $F_{\text{HPA}}(x(n))\approx F_{\text{HPA}}(\hat{x}(n|n-1))+f(n)(x(n)-\hat{x}(n|n-1))$ and by applying the conventional linear Kalman filter, we have the extended Kalman filter given by \cite[Ch. 13.7]{KAY}
\begin{align}
&\hat{x}(n|n-1)=a \hat{x}(n-1|n-1), \\
&M(n|n-1)=a^2 M(n-1|n-1)+\sigma_u^2,\\ 
&K(n)=\frac{M(n|n-1)\sqrt{\rho}h^*f(n)^*}{\sigma_e^2+\rho |h|^2 |f(n)|^2 M(n|n-1)}, \label{f22} \\
&\hat{x}(n|n)=\hat{x}(n|n-1)+K(n)[y'(n)\nonumber\\
&\qquad\qquad\qquad\qquad\qquad-\sqrt{\rho}h F_{\text{HPA}}(\hat{x}(n|n-1))], \label{f3}\\
&M(n|n)=(1-K(n)\sqrt{\rho}hf(n))M(n|n-1) \nonumber \\
&=\frac{a^2 \sigma_e^2 M(n-1|n-1)+\sigma_e^2 \sigma_u^2}{\sigma_e^2+\rho |h|^2 |f(n)|^2 (a^2 M(n-1|n-1)+\sigma_u^2)},\label{f4}
\end{align}
where $f(n)=\left. \frac{dF_\text{HPA}}{d x(n)} \right|_{x(n)=\hat{x}(n|n-1)}$ is given in Appendix B. The average harvested energy is given by
\begin{align}
\mathcal{E}(n) &=\zeta \mathbb{E}\{|y'(n)|^2\}=\zeta \mathbb{E}\{|\sqrt{1-\rho}h F_{\text{HPA}}(x(n))|^2\}\nonumber \\
&= \zeta (1-\rho)|h|^2 \mathbb{E}\{|F_{\text{HPA}}(x(n))|^2\}\},
\end{align}
where
\begin{align}
\mathbb{E}\{|F_{\text{HPA}}(x(n))|^2\} &= \int_0^\infty f_A(|x|)^2 \frac{2x}{\nu^2}\exp\!\left(\!-\frac{x^2\!+\!a^{2n+2}\mu_0^2}{\nu^2}\right)\nonumber\\
&\qquad\qquad\times I_0\left(\frac{2 x a^{n+1}\mu_0}{\nu^2}\right)dx,\label{hpa_rice}
\end{align}
where $f_A(|x|)$ is given by \eqref{char_fun}, $\nu^2 = a^{2n+2}\sigma_0^2 + \sigma^2(1-a^{2n+2})$, and $I_0(\cdot)$ is the modified Bessel function of the first kind. The above integral can be evaluated numerically but we can provide a closed-form for high values of $n$. In this case, the envelope of $x(n)$ follows a Rayleigh distribution with scale parameter $\sigma/\sqrt{2}$. Therefore, we have
\begin{align}
&\mathbb{E}\{|F_{\text{HPA}}(x(n))|^2 \}\!=\!\int_{0}^{\infty}\!\! \frac{x^2}{\left[ \left( \frac{x}{A_{\text{sat}}} \right)^{2\beta}\!+\!1 \right]^{\frac{1}{\beta}}}\frac{2x}{\sigma^2}\exp\left(-\frac{x^2}{\sigma^2} \right)dx\nonumber \\
&=\frac{A_\text{sat}^2}{\beta\sigma^2}\int_0^\infty\!\!\!\!\! \frac{z^{2/\beta-1}}{\left[z \!+\! A_{\text{sat}}^{2\beta} \right]^{\frac{1}{\beta}}}\exp\left(-\frac{z^{1/\beta}}{\sigma^2} \right)dz,
\end{align}
which follows from the transformation $x^{2\beta} \to z$. Then, by using the identity $\exp(x) = G^{1,0}_{0,1} \left(-x \big | \mfrac{-}{0} \right)$ and the result from \cite[2.24.2.4]{PB}, we have for $\beta \in \mathbb{Z}^+$
\begin{align}
\mathbb{E}\{|F_{\text{HPA}}(x(n))&|^2\} = \frac{A_\text{sat}^4 (2\pi)^{(1-\beta)/2}}{\sqrt{\beta}\sigma^2\Gamma(1/\beta)}\nonumber\\
&\times G_{1,\beta+1}^{\beta+1,1} \left( \left[\frac{A_\text{sat}^2}{\beta\sigma^2}\right]^\beta \bigg | \mfrac{1-2/\beta, -}{-1/\beta, \Delta(\beta)} \right),\label{hpa}
\end{align}
where $\Delta(\beta) = 0,1/\beta,\dots,(\beta-1)/\beta$ and $G_{p q}^{m n} \left(\cdot ~ \Big | ~ \mfrac{a_1,\dots, a_p}{b_1, \dots, b_q} \right)$ denotes the Meijer G-function \cite[9.301]{GRA}. Clearly, the HPA nonlinearities have a negative effect on the harvested energy, that is, \eqref{hpa_rice} and \eqref{hpa} are always less or equal to \eqref{var2} and $\sigma_u^2/(1-a^2)$, respectively. Note that equality is achieved for large values of $A_\text{sat}$ since then we have $f_A(|z|) \to |z|$.

\begin{figure}
\centering
\includegraphics[width=0.75\linewidth]{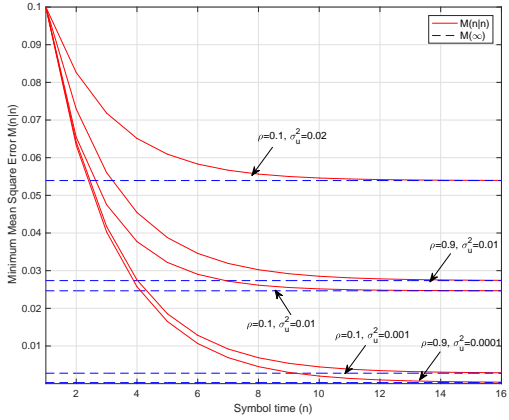}
\vspace{-0.15cm}
\caption{MMSE versus $n$; $a=0.8$, $\sigma_0^2=0.1$, $\mu_0=0$, $\sigma_u^2=\{0.0001, 0.001, 0.02 \}$, $\sigma_v^2=1$, $\sigma_q^2=0.5$, $\rho=\{0.1, 0.9\}$ and $|h|^2=1$.}\label{fig2}
\end{figure}

\begin{figure}
\centering
\includegraphics[width=0.75\linewidth]{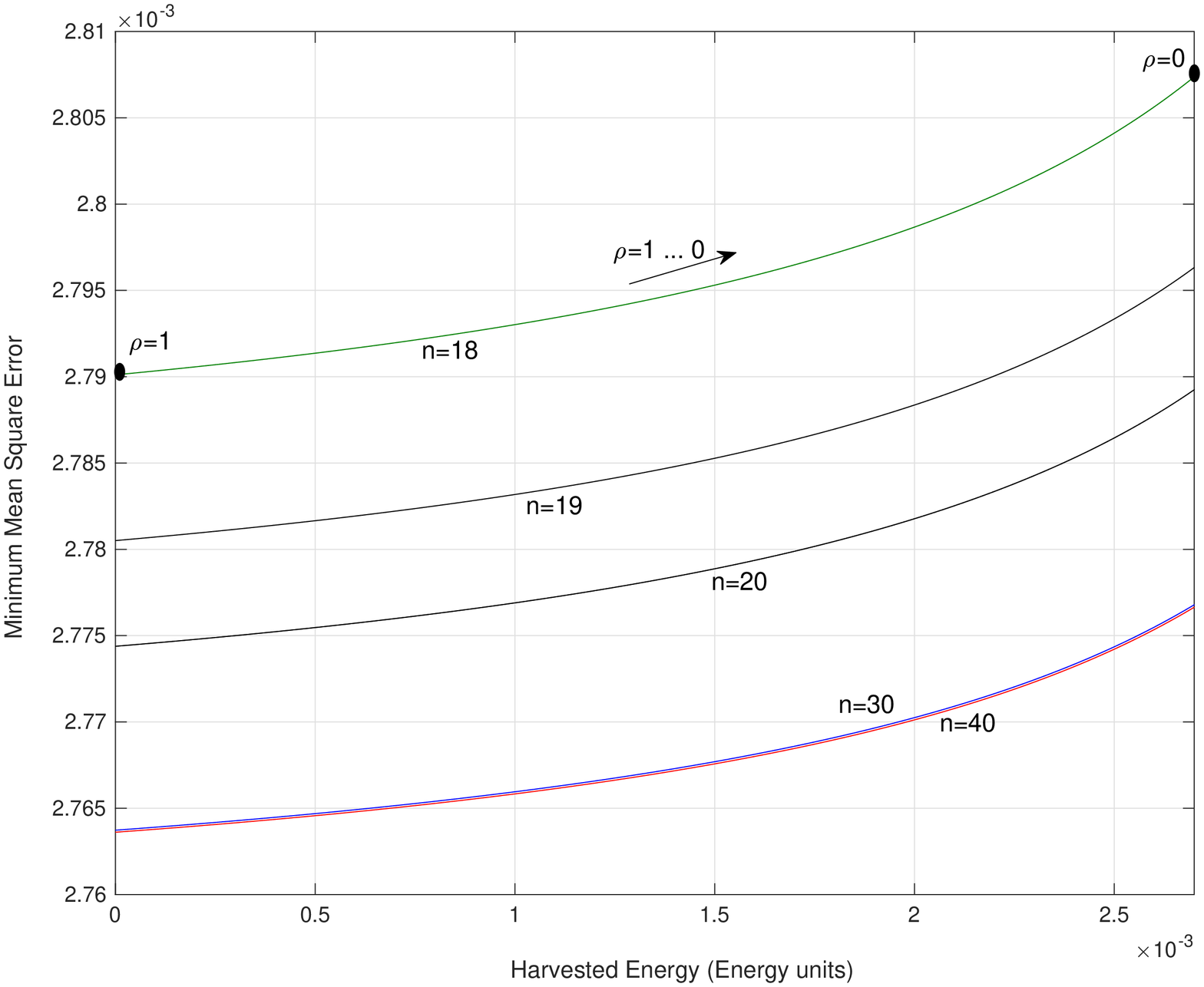}
\vspace{-0.2cm}
\caption{Estimation-energy tradeoff for various $n$; $\mu_0=0$, $\sigma_0^2=0.1$, $a=0.8$, $\sigma_u^2=0.001$, $\sigma_v^2=1$, $\sigma_q^2=0.5$, $|h|^2=1$, $\zeta=1$, $\rho\in [0,1]$.}\label{fig3}
\end{figure}

\section{Numerical results}
Computer simulations are carried-out to validate the performance of the proposed architecture. Fig. \ref{fig2} refers to the static channel scenario and plots the achieved MMSE 	$M(n|n)$ versus the symbol time $n$ for different parameters $\rho$ and $\sigma_u^2$. Since $\sigma_0^2>\sigma_u^2$, it can be seen that the MMSE performance is improved as $n$ increases. Asymptotically and according to the expression in \eqref{asymp1}, the MMSE performance converges to the floor $M(\infty)$. It can be also seen that the MMSE performance is improved as the parameter $\rho$ increases since a higher fraction of the received signal is used for state estimation. A similar trend is observed when the parameter $\sigma_u^2$ decreases; the excitation noise $\sigma_u^2$ is more critical for the estimation process since it determines the dynamics of the states and therefore significantly affects the performance of the linear Kalman filter. 

Fig. \ref{fig3} shows the MMSE-harvested energy tradeoff for various values of $n$ and for $0 \leq \rho \leq 1$ ($n$ affects the tradeoff when the system operates in the non asymptotic regime); each point of the curves corresponds to a different power splitting parameter $\rho$. 
It can be seen that as $n$ increases both estimation and harvesting performance are improved; however, the performance becomes independent of $n$ for $n>30$ (asymptotic regime).  In addition,  we observe that $\rho$ significantly affects the tradeoff since it defines the fraction of the received signal that is used for estimation and harvesting, respectively. The parameter $\rho$ is designed based on the system requirements to achieve a specific estimation-harvesting tradeoff.

Fig. \ref{fig4} deals with the fading case and presents the proposed CDF upper/lower bounds for different values of $a$; the empirical/true CDF is also plotted based on Monte Carlo simulations. As it can be seen, the proposed bounds efficiently approximate the true CDF for all the values of $a$; in addition, it can be seen that as $a$ decreases the CDF curves are shifted to the left ( i.e., the MMSE performance degrades since the correlation of consecutive states decreases) and their gap is squeezed. Fig. \ref{fig44} plots the average MMSE versus the harvested energy for a fading scenario with small values of $a$ and large values of $n$ (asymptotic case); for different values of $\rho$, the proposed architecture achieves a different estimation-energy tradeoff. A comparison with the static case, shows that channel fading boosts the MMSE performance due to the associated time diversity. Finally, in Fig. \ref{fig5}, we study the impact of HPA and plot the MMSE performance versus the symbol time $n$ (we assume $\sigma_0^2=\sigma_u^2<\sigma^2$ and thus the MMSE increases as $n$ increases).  It can be seen that HPA significantly affects the MMSE performance since it makes the observation equation nonlinear and degrades the performance of the linear Kalman filter (extended Kalman).

\begin{figure}
\centering
\includegraphics[width=0.77\linewidth]{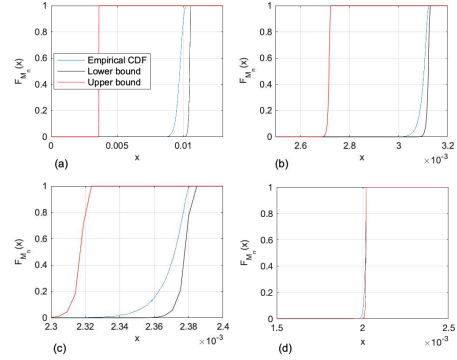}
\vspace{-0.3cm}
\caption{CDF of MMSE for the fading scenario- empirical and upper/lower bounds; $n=10$, $\rho=0.9$, $\mu_0=0$, $\sigma_0^2=0.1$, $\sigma_u^2=0.002$, $\sigma_v^2=1$, $\sigma_q^2=0.5$, $\lambda=1$; (a) $a=0.9$, (b) $a=0.6$, (c) $a=0.4$, and (d) $a=0.1$.}\label{fig4}
\end{figure}

\begin{figure}
\centering
\includegraphics[width=0.7\linewidth]{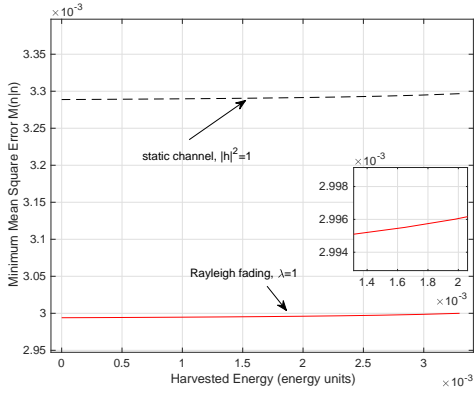}
\vspace{-0.2cm}
\caption{MMSE versus energy harvesting for a fading scenario (asymptotic case) with $\lambda=1$, $a=0.3$, $\zeta=1$, $\sigma_u^2=0.003$, $\sigma_v^2=1$, $\sigma_q^2=0.5$.}\label{fig44}
\end{figure}

\begin{figure}
\centering
\includegraphics[width=0.7\linewidth]{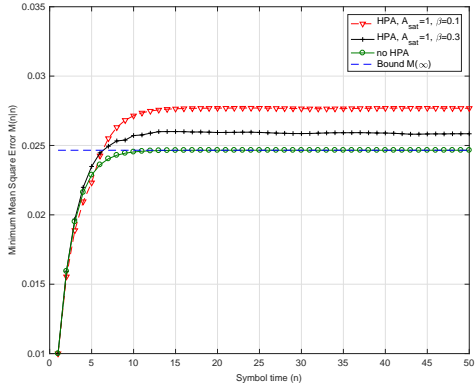}
\vspace{-0.2cm}
\caption{MMSE versus the symbol time $n$ for the scenario with HPA; $\rho=0.9$, $a=0.8$, $\mu_0=0$, $\sigma_0^2=0.01$, $\sigma_u^2=0.01$, $\sigma_v^2=1$, $\sigma_q^2=0.5$, and $|h|^2=1$.}\label{fig5}
\end{figure}

\vspace{-0.2cm}
\section{Conclusion}

This paper has studied a receiver architecture which enables simultaneous estimation of a Gauss-Markov signal and WPT. We have considered three operation scenarios: static channel, Rayleigh block-fading, and transmission with HPA. The considered architecture employs a Kalman filter for estimation and is characterized by a fundamental tradeoff between estimation quality (MMSE) and average harvested energy. Closed-form expressions are given for both MSE performance and average harvested energy. It has been shown that channel fading enhances MSE while HPA significantly affects both estimation/harvesting performance.

\appendices 

\vspace{-0.2cm}
\section{Proof of Theorem 1}

An upper bound corresponds to the case where, due to the channel fading, no observation is available at the receiver. In this case, the MMSE is upper bounded by $M(n|n-1)\leq \text{variance}(x(n))\rightarrow \frac{\sigma_u^2}{1-a^2}$ (for high $n$) which is similar to the initialization phase of the Kalman filter \cite{KAY}. Regarding the lower bound, from the expressions in \eqref{e2}, we have $M(n|n-1)\geq \sigma_u^2$ since $M(n-1|n-1)\geq 0$.

Since the MMSE $M(n|n-1)$ is bounded in both cases by a constant, the random variable $M(n|n)$ is simplified as 
\begin{align}
M(n|n)=\frac{\sigma_e^2}{\sigma_e^2/c+\rho |h(n)|^2},
\end{align}
where $c=a^2 c_b+\sigma_u^2$ is a constant and $c_b$ denotes the upper/lower bound. The CDF is given by 
\begin{align}
F_{M_n}(x)&=\mathbb{P} \left \{\frac{\sigma_e^2}{\sigma_e^2/c+\rho |h(n)|^2} \leq x \right \} \nonumber \\
&=\mathbb{P} \left \{|h(n)|^2\geq \frac{\sigma_e^2}{\rho}\left[ \frac{1}{x}-\frac{1}{c}\right] \right\} \nonumber \\
&=\left\{ \begin{array}{ll} \exp \left(-\frac{\lambda \sigma_e^2}{\rho}\left(\frac{1}{x}-\frac{1}{c} \right)\right),\;\;\textrm{if}\;x<c, \\
1,\;\;\textrm{elsewhere.}\end{array}\right. 
\end{align}
For the upper bound, we have $c_b=\sigma_u^2/(1-a^2)$ and thus $c=c_b=\sigma_u^2/(1-a^2)$; for the lower bound $c_b=\sigma_u^2$ and thus $c=\sigma_u^2(1+a^2)$. This completes the proof. 

\vspace{-0.2cm}
\section{Extended Kalman - Computation of $f(n)$}

Let $z=x+jy$ be the complex input to the power amplifier. Then, the output is given by 
\begin{align}
F_{\text{HPA}}(z)&=f_A(|z|)\exp(j \arg{z}) \nonumber \\
&= \frac{\sqrt{x^2+y^2}}{\left[ \left( \frac{\sqrt{x^2+y^2}}{A_{\text{sat}}} \right)^{2\beta}+1 \right]^{\frac{1}{2\beta}}} \cos \left(\tan^{-1}\frac{y}{x} \right ) \nonumber \\
&+j \frac{\sqrt{x^2+y^2}}{\left[ \left( \frac{\sqrt{x^2+y^2}}{A_{\text{sat}}} \right)^{2\beta}+1 \right]^{\frac{1}{2\beta}}} \sin \left(\tan^{-1}\frac{y}{x} \right ) \nonumber \\
&= U(x,y)+jV(x,y).
\end{align}
By applying the Cauchy Riemann equations \cite[11.2]{CR},
{\scriptsize
\begin{align}
&f(n)=\left. \frac{dF_{\text{HPA}}}{dz} \right|_{z=x_0+j y_0}=\frac{dU}{dx}(x_0,y_0)+j\frac{dV}{dx}(x_0,y_0) \nonumber \\
&= \frac{\sqrt{x_0^2+y_0^2} \left(1+\left(\frac{x_0^2+y_0^2}{A_{\text{sat}}^2}\right)^{\beta}\right)^{-\frac{1+2\beta}{2\beta}} \left(x_0^2+y_0^2 \left(1+\left(\frac{x_0^2+y_0^2}{A_{\text{sat}}^2}\right)^{\beta}\right)\right)}{\left(x_0^2+y_0^2\right)^{3/2}} \nonumber \\
&\qquad-j \frac{ y_0 \left(\frac{\sqrt{x_0^2+y_0^2}}{A_{\text{sat}}}\right)^{2\beta-1} \left(1+\left(\frac{x_0^2+y_0^2}{A^2_{\text{sat}}}\right)^{\beta}\right)^{-\frac{1+2\beta}{2\beta}}}{A_{\text{sat}} \sqrt{1+\frac{y_0^2}{x_0^2}}},
\end{align}}
where $z=\hat{x}(n|n-1)$, $x_0=\Re\{\hat{x}(n|n-1)\}$ and $y_0=\Im\{\hat{x}(n|n-1)\}$.

\end{document}